\pdfoutput=1
\documentclass[10pt,a4paper]{article}

\usepackage[T1]{fontenc}
\usepackage[utf8]{inputenc}
\usepackage[english]{babel}
\usepackage[a4paper,margin=1in]{geometry}

\usepackage[fleqn]{amsmath}
\usepackage{amssymb}
\usepackage{amsthm}
\usepackage{mathtools}
\usepackage{thm-restate}

\usepackage[linesnumbered,ruled,noend,vlined]{algorithm2e}

\usepackage[bookmarksnumbered]{hyperref}

\hypersetup{
  hidelinks,
  pdftitle={Linear-Query Deterministic Approximation for Non-monotone Submodular Maximization under a Knapsack Constraint},
  pdfauthor={Zihui Liu and Zhijie Zhang},
  pdfkeywords={Submodular maximization, knapsack constraint, linear-time algorithms, bicriteria optimization}
}

\theoremstyle{plain}
\newtheorem{theorem}{Theorem}
\newtheorem{lemma}[theorem]{Lemma}
\newtheorem{observation}[theorem]{Observation}
\newtheorem{claim}[theorem]{Claim}
\newtheorem{definition}[theorem]{Definition}

\title{Linear-Query Deterministic Approximation for Non-monotone Submodular Maximization under a Knapsack Constraint}

\author{%
Zihui Liu\\[0.7ex]
Fuzhou University\\
\texttt{240320049@fzu.edu.cn}%
\and
Zhijie Zhang\\[0.7ex]
Fuzhou University\\
\texttt{zzhang@fzu.edu.cn}%
}

\date{}

\begin{document}

\maketitle

\begin{abstract}
Submodular maximization under a knapsack constraint (SMK) is a fundamental combinatorial optimization problem with broad applications across machine learning and data mining.
Motivated by large-scale applications where query efficiency is paramount, we study non-monotone SMK and focus on deterministic algorithms with linear query complexity.
Prior deterministic linear-query algorithms achieve at best a $1/5-\varepsilon$ approximation, falling short of the $1/4-\varepsilon$ ratio attainable by randomized algorithms.
We close this gap by presenting a deterministic $(1/4-\varepsilon)$-approximation with $O(n\log^2(1/\varepsilon)/\varepsilon^2)$ queries.
Our approach partitions the analysis based on the cost of the largest optimal element $r$: when the cost of $r$ is moderate, we refine the threshold-twin-greedy framework via residual-budget enumeration to tighten the analysis; when the cost of $r$ is large, we reduce the problem to bicriteria submodular maximization.
As a secondary contribution, we obtain a $(1/2-\varepsilon, O(1/\varepsilon))$-bicriteria approximation with $O(n\log(1/\varepsilon)/\varepsilon^2)$ queries, improving over the previous $O(n^2/\varepsilon)$ query bound.
\end{abstract}

\noindent\textbf{Keywords}\enskip Submodular maximization, knapsack constraint, linear-time algorithms, bicriteria optimization

\section{Introduction}
\label{sec:intro}

Submodular functions capture the principle of diminishing returns and have found widespread applications across machine learning and combinatorial optimization. Classical use cases include influence maximization in social networks~\cite{KempeKT03}, sensor placement and outbreak detection~\cite{LeskovecKGFVG07}, document and image summarization~\cite{LinB11,MirzasoleimanKS16}, active learning and data subset selection~\cite{NarasimhanB05,BadanidiyuruV14}, revenue maximization and recommendation systems~\cite{Kuhnle21,PhamT0T23,Pham25}, facility location and coverage problems~\cite{Wolsey82,KhullerMN99}, and canonical graph problems such as maximum cut~\cite{GuptaRST10,AmanatidisFLLR20}.

In this work, we study \emph{submodular maximization under a knapsack constraint (SMK)}. Given a finite ground set \(V\), a non-negative submodular function \(f:2^V\to\mathbb{R}_+\) accessed via a value oracle, a positive cost function \(c:V\to\mathbb{R}_+\), and a budget \(B\), the goal is to find a set \(S\subseteq V\) maximizing \(f(S)\) subject to \(c(S)\le B\).

When \(f\) is monotone, Nemhauser and Wolsey~\cite{NemhauserW78} established that any algorithm achieving an approximation ratio better than \(1-1/e\) requires exponentially many oracle queries. The classical \((1-1/e)\)-approximation was attained by Sviridenko~\cite{Sviridenko04} with \(O(n^5)\) queries and later improved to \(O(n^4)\) by Feldman, Nutov, and Shoham~\cite{FeldmanNS23}. For the non-monotone setting, Qi~\cite{Qi22} proved that any algorithm with ratio better than \(0.478\) requires exponential queries. Building on a long line of work~\cite{KulikST13,BuchbinderF19,BuchbinderF24}, the state-of-the-art polynomial-query algorithm achieves a \(0.401\)-approximation~\cite{BuchbinderF24}.

Despite their strong theoretical guarantees, these algorithms suffer from prohibitively high query complexity, rendering them impractical for large-scale datasets. This has motivated a growing body of work on \emph{fast} SMK algorithms whose query complexity is (nearly) linear in \(n=|V|\), often trading off approximation ratio for efficiency.

For monotone objectives, building upon Badanidiyuru and Vondr{\'{a}}k~\cite{BadanidiyuruV14}, Ene and Nguyen \cite{EneN19} proposed a \((1-1/e-\varepsilon)\)-approximation using \((1/\varepsilon)^{O(1/\varepsilon^4)} n\log^2 n\) queries. While theoretically near-optimal, the towering dependence on \(1/\varepsilon\) limits practicality. Subsequent work pursued more practical algorithms with weaker guarantees. Specifically, Li et al.~\cite{LiF0K22} designed a \((1/2-\varepsilon)\)-approximation with \(O(n\log(1/\varepsilon)/\varepsilon)\) queries.

For non-monotone objectives, the landscape of linear-query algorithms has evolved rapidly.
On the randomized side, Amanatidis et al.~\cite{AmanatidisFLLR20} achieved a \((0.1715-\varepsilon)\)-approximation with \(O(n\log(n/\varepsilon)/\varepsilon)\) queries.
Han et al.~\cite{HanCZZWYXTH21} gave a \((1/4-\varepsilon)\)-approximation with \(O(n\log(n/\varepsilon)/\varepsilon)\) queries.
Pham et al.~\cite{PhamT0T23} later matched this ratio with \(O(n\log(1/\varepsilon)/\varepsilon)\) queries. On the deterministic side, Han et al.~\cite{HanCZZWYXTH21} achieved a \((1/6-\varepsilon)\)-approximation with \(O(n\log(n/\varepsilon)/\varepsilon)\) queries. Pham et al.~\cite{PhamT0T23} introduced the first deterministic linear-query algorithm for non-monotone SMK, achieving the same \((1/6-\varepsilon)\) guarantee with \(O(n\log(1/\varepsilon)/\varepsilon)\) queries. Very recently, Pham~\cite{Pham25} improved the ratio to \(1/5-\varepsilon\) while preserving linear query complexity.

Notably, however, the best approximation achievable by deterministic linear-query algorithms remains strictly below that of their randomized counterparts. The inherent gap between deterministic and randomized approaches has been a central theme throughout the study of submodular maximization~\cite{BuchbinderF18,BuchbinderF19soda,BuchbinderF24focs,BuchbinderF25}. This state of affairs naturally motivates the following question:
\begin{center}
\emph{Can a deterministic algorithm with linear queries match the approximation factor as the best randomized algorithm for non-monotone SMK?}
\end{center}

\subsection{Our Contributions}

We answer the above question affirmatively by presenting a deterministic linear-query algorithm for non-monotone SMK that matches the \(1/4-\varepsilon\) approximation ratio previously achievable only by randomized algorithms. Table~\ref{tab:smk-comparison} summarizes prior and our results.

\begin{table}[t]
\centering
\caption{Algorithms for non-monotone SMK with (nearly) linear query complexity.}
\label{tab:smk-comparison}
\begin{tabular}{lccc}
\hline
Algorithm & Type & Approximation factor & Query complexity \\
\hline
\multicolumn{4}{c}{\emph{Randomized}} \\
\hline
Amanatidis et al.~\cite{AmanatidisFLLR20} & Rand. & \(0.1715-\varepsilon\) & \(O(n\log(n/\varepsilon)/\varepsilon)\) \\
Han et al.~\cite{HanCZZWYXTH21} & Rand. & \(1/4-\varepsilon\) & \(O(n\log(n/\varepsilon)/\varepsilon)\) \\
Pham et al.~\cite{PhamT0T23} & Rand. & \(1/4-\varepsilon\) & \(O(n\log(1/\varepsilon)/\varepsilon)\) \\
\hline
\multicolumn{4}{c}{\emph{Deterministic}} \\
\hline
Han et al.~\cite{HanCZZWYXTH21} & Det. & \(1/6-\varepsilon\) & \(O(n\log(n/\varepsilon)/\varepsilon)\) \\
Pham et al.~\cite{PhamT0T23} & Det. & \(1/6-\varepsilon\) & \(O(n\log(1/\varepsilon)/\varepsilon)\) \\
Pham~\cite{Pham25} & Det. & \(1/5-\varepsilon\) & \(O(n\log(1/\varepsilon)/\varepsilon)\) \\
\textbf{This paper} & Det. & \(1/4-\varepsilon\) & \(O(n\log^2(1/\varepsilon)/\varepsilon^2)\) \\
\hline
\end{tabular}
\end{table}

\begin{restatable}{theorem}{ThmMain}
\label{thm:main}
There exists a deterministic algorithm that achieves an approximation ratio of \(1/4-\varepsilon\) for non-monotone SMK using \(O(n\log^2(1/\varepsilon)/\varepsilon^2)\) value oracle queries.
\end{restatable}

Our approach distinguishes two cases based on the cost of the largest element $r$ in an optimal solution $O$.
In the regime where $c(r) \le (1-\varepsilon)B$, we enhance the threshold-twin-greedy framework~\cite{PhamT0T23, Pham25} via residual-budget enumeration: we guess $B - c(r)$ and run the framework independently for each guess.
While this incurs a multiplicative $O(\log(1/\varepsilon)/\varepsilon)$ overhead in query complexity, it enables a tighter analysis that improves the approximation ratio from $1/5$ to $1/4$.

When $c(r) \ge (1-\varepsilon)B$, we observe that this case reduces to \emph{bicriteria submodular maximization}~\cite{Crawford23,FeldmanK25}: given budget \(B\) and parameters \(\alpha\in(0,1)\) and \(\beta\ge 1\), an algorithm is an \((\alpha,\beta)\)-bicriteria approximation if it returns a set \(S\) with \(f(S)\ge\alpha\cdot f(O)\) and \(c(S)\le\beta\cdot B\).
Crawford~\cite{Crawford23} gave a \((1/2-\varepsilon, O(\varepsilon^{-2}))\)-bicriteria approximation using \(O(n^2/\varepsilon^3)\) queries.
Feldman and Kuhnle~\cite{FeldmanK25} achieved a \((1/2-\varepsilon, O(1/\varepsilon))\)-bicriteria approximation using \(O(n^2/\varepsilon)\) queries.
We improve on these results in Theorem~\ref{thm:bicriteria-main}.
The first item yields a $(1/3-\varepsilon, 12)$-bicriteria approximation with $O(n\log(1/\varepsilon)/\varepsilon)$ queries, which is sufficient for handling the case $c(r) \ge (1-\varepsilon)B$.
\begin{restatable}{theorem}{ThmBicriteria}
    \label{thm:bicriteria-main}
    For bicriteria SMK,
    \begin{enumerate}
        \item there exists a deterministic \((1/3-\varepsilon, 12)\)-bicriteria approximation algorithm for SMK that uses \(O(n\log(1/\varepsilon)/\varepsilon)\) queries;
        \item there exists a deterministic \((1/2-\varepsilon, O(1/\varepsilon))\)-bicriteria approximation algorithm for SMK that uses \(O(n\log(1/\varepsilon)/\varepsilon^2)\) queries.
    \end{enumerate}
\end{restatable}

\begin{table}[t]
\centering
\caption{Bicriteria approximation algorithms for non-monotone SMK.}
\label{tab:bicriteria-comparison}
\begin{tabular}{lcccc}
\hline
Algorithm & Type & Ratio (\(\alpha\)) & Budget violation (\(\beta\)) & Query complexity \\
\hline
Crawford~\cite{Crawford23} & Det. & \(1/2-\varepsilon\) & \(O(1/\varepsilon^{2})\) & \(O(n^2/\varepsilon^3)\) \\
Feldman \& Kuhnle~\cite{FeldmanK25} & Det. & \(1/2-\varepsilon\) & \(O(1/\varepsilon)\) & \(O(n^2/\varepsilon)\) \\
\textbf{This paper} & Det. & \(1/3-\varepsilon\) & \(12\) & \(O(n\log(1/\varepsilon)/\varepsilon)\) \\
\textbf{This paper} & Det. & \(1/2-\varepsilon\) & \(O(1/\varepsilon)\) & \(O(n\log(1/\varepsilon)/\varepsilon^2)\) \\
\hline
\end{tabular}
\end{table}

\subsection{Paper Organization}

The remainder of this paper is organized as follows. Section~\ref{sec:pre} introduces preliminaries and notation. Section~\ref{sec:small_c(r)} handles the case where the largest element in the optimum has a moderate cost. Section~\ref{sec:large_c(r)} addresses the complementary case. Section~\ref{sec:conclusion} concludes.

\section{Preliminaries}
\label{sec:pre}

Let \(V\) be a finite ground set of size \(n\).
For a set \(S\subseteq V\) and an element \(e\in V\), we write \(S+e\) for \(S\cup\{e\}\) and \(S-e\) for \(S\setminus\{e\}\), respectively.
The marginal contribution of an element \(e\) to a set \(S\) is denoted by \(f(e\mid S)= f(S+e)-f(S)\), and the marginal contribution of a set \(T\) to \(S\) is denoted by \(f(T\mid S)= f(S\cup T)-f(S)\).
A set function \(f:2^V\to\mathbb R\) is \emph{submodular} if for every \(S\subseteq T\subseteq V\) and every \(e\in V\setminus T\),
\[
        f(S+e)-f(S)\ge f(T+e)-f(T).
\]
Equivalently, \(f\) satisfies \(f(S)+f(T)\ge f(S\cup T)+f(S\cap T)\) for all \(S,T\subseteq V\).
The function is \emph{monotone} if \(f(S) \le f(T)\) whenever \(S \subseteq T\), and \emph{non-negative} if \(f(S) \ge 0\) for all \(S \subseteq V\).
Each element \(e\in V\) has a positive cost \(c(e)>0\), and for a set \(S\subseteq V\) we define \(c(S)=\sum_{e\in S}c(e)\).
A set \(S\) is \emph{feasible} if \(c(S) \le B\), where \(B > 0\) is the given budget.
We assume w.l.o.g.~that \(c(e)\le B\) for every \(e\in V\), since otherwise the element is infeasible and can be safely discarded.

The problem of \emph{submodular maximization under a knapsack constraint (SMK)} can be formulated as
\[ \max\{f(S): S\subseteq V,\ c(S)\le B\}. \]
Here \(f\) is non-negative and submodular, but not necessarily monotone.
We assume that \(f\) is accessed through a value oracle that returns \(f(S)\) when a set \(S\subseteq V\) is queried.
We measure the efficiency of the algorithm by query complexity, i.e., the number of oracle queries it makes.

The algorithm presented in this paper invokes two subroutines.
The first is a linear-query algorithm that estimates the optimal value of SMK within a constant factor.

\begin{lemma}[\cite{PhamT0T23}]
\label{lem:LA-estimate}
 For SMK, there is a deterministic algorithm \emph{LA} that uses \(O(n)\) queries and returns a value \(\Gamma\) satisfying \(\Gamma\leq \max_{T\subseteq V, c(T)\leq B} f(T)\leq 19\Gamma\).
\end{lemma}

The second subroutine is a linear-query algorithm for solving the problem of \emph{unconstrained submodular maximization (USM)}, namely \(\max\{f(S): S\subseteq V\}\).

\begin{lemma}[\cite{BuchbinderF18}]
    \label{lem:USM}
    For USM, there exists a deterministic algorithm \emph{DetUSM} that uses \(O(n/\varepsilon)\) queries and returns a set \(S\subseteq V\) that satisfies \(f(S)\geq (1/4)\cdot f(\varnothing)+(1/2-\varepsilon)\cdot \max_{T\subseteq V} f(T)\).
\end{lemma}

Let \(O\) be an optimal solution of SMK.
Let \(r\in\arg\max_{o\in O}c(o)\) be an element with the highest cost in \(O\), and define \(O'=O-r\).
Given the parameter \(\varepsilon>0\), our algorithm splits into two parts, depending on whether \(c(r)\) exceeds \((1-\varepsilon)B\).
We introduce them in Sections \ref{sec:small_c(r)} and \ref{sec:large_c(r)}, respectively.

\section{The Algorithm for \texorpdfstring{\(c(r)\leq (1-\varepsilon)B\)}{c(r) <= (1-epsilon)B}}
\label{sec:small_c(r)}

In this section, we present a linear-query deterministic algorithm for SMK that achieves a \((1/4-\varepsilon)\)-approximation, assuming that \(c(r)\leq (1-\varepsilon)B\).

The algorithm is presented as Algorithm~\ref{alg:small-r-fixed}.
It first guesses the residual budget \(B-c(r)\) via a geometric sequence \(\varepsilon (1+\varepsilon)^{i} B\) for \(i = 0,1,\ldots, \left\lfloor\log_{1+\varepsilon}\varepsilon^{-1}\right\rfloor\).
For each guess \(i\), the algorithm maintains two disjoint partial solutions \(X\) and \(Y\), initialized as empty sets, along with their \emph{saturation flags} \(\delta_X = \delta_Y = 0\).
It then executes two phases: a threshold-search phase and an augmentation phase.

In the threshold-search phase, the algorithm iteratively lowers a density threshold \(\tau\), starting from \(19\Gamma/(4\varepsilon B)\) and updating \(\tau \gets (1-\varepsilon)\tau\) until \(\tau\leq \varepsilon(1-\varepsilon)\cdot \Gamma/B\).
For each value of \(\tau\), the algorithm scans all unselected elements \(e \in V' \setminus (X \cup Y)\), where \(V'=\{e \in V: c(e) \le (1-\varepsilon(1+\varepsilon)^i)B\}\), and attempts to assign \(e\) to one of \(\{X, Y\}\).
For each element \(e\), it identifies a set \(Z \in \mathcal{C} = \{X, Y\}\) satisfying \(f(e \mid Z)/c(e) \ge \tau\), breaking ties by selecting the set maximizing the marginal density. If such a set exists, the algorithm checks whether adding \(e\) would cause \(Z\) to exceed the guessed residual budget, i.e.~\(c(Z+e) > \varepsilon(1+\varepsilon)^i B\). If so, \(Z\) is removed from the candidate pool \(\mathcal{C}\), its saturation flag \(\delta_Z\) is set to \(1\), and immediately considers the same
element for the remaining set in \(\mathcal C\); otherwise, \(e\) is added to \(Z\). Once a set is saturated, it no longer participates in the threshold search but is retained for the final augmentation phase.

After the threshold loop terminates, each candidate solution undergoes a final augmentation.
For \(Z \in \{X, Y\}\), the algorithm identifies an element \(e^*=\arg\max_{e \in V',\, f(e \mid Z) \ge 0} f(e \mid Z)\)
(if no such element exists, \(e^* = \varnothing\)) and forms \(Z^+ = Z + e^*\).
The global solution \(S^*\) is updated to retain the maximum among \(f(S^*)\), \(f(X^+)\), and \(f(Y^+)\).

\begin{algorithm}[tb]
\SetAlgoLined
\DontPrintSemicolon
\caption{Threshold Twin Greedy \(+\) Augmentation}
\label{alg:small-r-fixed}
\KwIn{Instance \((V,f,c,B)\), value \(\Gamma\) satisfying \(\Gamma\le f(O)\le 19\Gamma\), parameter \(\varepsilon\).}
\(S^*\gets \varnothing\).

\For{\(i=0\) \KwTo \(\lfloor\log_{1+\varepsilon}\varepsilon^{-1}\rfloor\)}{
    Let \(X,Y\gets\varnothing\), \(\mathcal{C}\gets\{X,Y\}\), and \(\delta_X=\delta_Y=0\).

    Let \(\tau\gets 19\Gamma/(4\varepsilon B)\) and \(V'\gets \{e\in V: c(e)\leq \left(1-\varepsilon(1+\varepsilon)^{i}\right)B\}\).
	
	\While{\(\tau>\varepsilon(1-\varepsilon)\cdot \Gamma/B\)}{
		\ForEach{\(e\in V'\setminus(X\cup Y)\)}{
			\ForEach{\(Z\in\mathcal{C}\) in nonincreasing order of
    \(f(e\mid Z)/c(e)\)\label{line:find_Z}}{
    \If{\(f(e\mid Z)/c(e)\geq\tau\)}{
        \eIf{\(c(Z+e)>\varepsilon(1+\varepsilon)^iB\)}{
            \(\mathcal{C}\gets\mathcal{C}\setminus\{Z\}\)
            and \(\delta_Z\gets1\).\;
        }{
            \(Z\leftarrow Z+e\).\;
            \textbf{break}\;
        }
    }
}
		}
		Update \(\tau\gets(1-\varepsilon)\tau\).
	}
    \ForEach{\(Z\in\{X,Y\}\)}{
        Let \(e^* \gets \arg\max_{e \in V',\; f(e \mid Z)\geq 0} f(e \mid Z)\). \tcp*{if no such \(e\), \(e^{*}=\varnothing\)}
            
        Let \(Z^{+} \gets Z+e^*\).
    }
    \(S^*\gets \arg\max\{f(S^*), f(X^{+}), f(Y^{+})\}\).
}
\Return \(S^*\).
\end{algorithm}

\subsection{The Analysis}

Recall that we denote by \(O\) an optimal solution, \(r=\arg\max\{c(o):o\in O\}\) and \(O'=O-r\).
We make two remarks about the algorithm.
First, under the condition \(c(r)\leq (1-\varepsilon)B\), there exists an index \(i\in\{0,1,\ldots,\lfloor\log_{1+\varepsilon}\varepsilon^{-1}\rfloor\}\) such that
\begin{equation}
    \varepsilon(1+\varepsilon)^{i}B\leq B-c(r)< \varepsilon(1+\varepsilon)^{i+1}B. \label{eq:guess-i}
\end{equation}
This follows from two observations: (i) \(c(r) \le (1-\varepsilon)B\) implies \(B - c(r) \ge \varepsilon B\), so \(i \ge 0\); (ii) if \(i > \log_{1+\varepsilon}\varepsilon^{-1}\) then \(\varepsilon(1+\varepsilon)^i B > B\), contradicting \(B - c(r) < B\). Hence \(i \le \log_{1+\varepsilon}\varepsilon^{-1}\), and since \(i\) is integral, \(0 \le i \le \lfloor\log_{1+\varepsilon}\varepsilon^{-1}\rfloor\).
Second, for every \(e \in O\), we have \(c(e) \le c(r) \le (1-\varepsilon(1+\varepsilon)^i) B\), and thus \(O \subseteq V'\).
Consequently, it suffices to consider only the elements in \(V'\) without loss of optimality.
In what follows, we restrict our analysis to the pair of sets \(X, Y\) constructed for the guess of \(i\) satisfying Eq.~\eqref{eq:guess-i}.

We now establish the notation.
Let \(X, Y, \delta_X, \delta_Y\) denote their final values.
Suppose \(X\) contains \(p\) elements and \(Y\) contains \(q\) elements, and define \(X_i = \{x_1,\ldots,x_i\}\) and \(Y_i = \{y_1,\ldots,y_i\}\) as the prefixes consisting of the first \(i\) elements added to \(X\) and \(Y\), respectively, so that \(X = X_p\) and \(Y = Y_q\).
If \(\delta_X = 1\), the algorithm attempted to add an element \(x_{p+1}\) to \(X_p\) but failed because \(c(X_p \cup \{x_{p+1}\}) > \varepsilon(1+\varepsilon)^i B\). We define \(X_{p+1} = X_p \cup \{x_{p+1}\}\); note that \(x_{p+1} \in V'\), so \(X_{p+1}\) is a valid candidate in the augmentation stage and therefore also a candidate for \(X^+\) and \(S^*\).
The symmetric definitions for \(\delta_Y = 1\) are analogous.

For each \(e\in X_p\cup Y_q\), let \(X_e\) and \(Y_e\)
denote the states of \(X\) and \(Y\), respectively,
immediately before \(e\) was inserted into its chosen
solution.
If \(\delta_X=1\), additionally set
\(X_{x_{p+1}}=X_p\); if \(\delta_Y=1\), additionally
set \(Y_{y_{q+1}}=Y_q\).

Finally, let \(Y'_p\) be the state of \(Y\) immediately before \(\delta_X\) was set to \(1\); if \(\delta_X = 0\) throughout, we set \(Y'_p = Y_q\).
The definition of \(X'_q\) is analogous.

We begin the analysis with two elementary observations.
Since \(f\) is not monotone, neither \(f(O \cup X_p)\) nor \(f(O \cup Y_q)\) is guaranteed to exceed \(f(O)\).
Nevertheless, Observation~\ref{obs:upper-bound-of-f(O)} establishes that their sum dominates \(f(O)\), enabling us to circumvent the technical difficulties arising from the lack of monotonicity.

\begin{observation}
\label{obs:upper-bound-of-f(O)}
\(f(O)\leq f(O\cup X_p)+f(O\cup Y_q)\).
\end{observation}

\begin{proof}
Note that \(X_p\cap Y_q=\varnothing\).
By submodularity, \(f(O\cup X_p)+f(O\cup Y_q)\geq f(O\cup X_p\cup Y_q)+f(O)\).
Since \(f\) is non-negative, \(f(O\cup X_p\cup Y_q)\geq 0\). 
The observation follows.
\end{proof}

To establish the desired approximation ratio, we must bound \(f(O \cup S) - f(S)\) from above for the feasible candidates \(S \in \{X_p, Y_q\}\).
using the fact that \(X_p + r\) and \(Y_q + r\) are candidate solutions during augmentation (recall that \(c(X_p + r), c(Y_q + r)\leq B\)), Observation \ref{obs:S+r} allows us to reduce the problem to bounding \(f(O' \cup S) - f(S)\), which is substantially more tractable.

\begin{observation}
\label{obs:S+r}
For any \(S\subseteq V\), \(f(O\cup S)-f(S+r)\leq f(O'\cup S)-f(S)\).
\end{observation}

\begin{proof}
If \(r\in S\), then \(O\cup S=O'\cup S\) and \(S+r=S\), and hence both sides are equal.
If \(r\notin S\), since \(S\subseteq S+r\), the submodularity gives
\[
        f(O\cup S)-f(S+r)=f(O'\cup S+r)-f(S+r)\leq f(O'\cup S)-f(S).\qedhere
\]
\end{proof}

Next, we bound \(f(O'\cup S)-f(S)\) for \(S \in \{X_p, Y_q\}\).
We focus on the case \(S = X_p\); the argument for \(Y_q\) is symmetric.
By submodularity, \(f\left( O^{\prime} \cup X_p\right)-f\left( X_p\right) \leq \sum_{e \in O^{\prime} \setminus X_p} f\left( e \mid X_p\right)\).
To bound this sum, we partition \(O' \setminus X_p\) into two disjoint subsets: \(O' \cap Y'_p\) and \(O' \setminus (X_p \cup Y'_p)\), and bound each part separately.
Lemma~\ref{lem:aux1} handles the sum over \(O^{\prime} \cap Y'_p\).
Lemmas~\ref{lem:not_added} and \ref{lem:added_is_larger} address the sum over \(e \in O^{\prime} \setminus\left(X_p \cup Y'_p\right)\), depending on whether \(X\) becomes saturated: the former applies when \(\delta_X = 0\), and the latter when \(\delta_X = 1\).

\begin{lemma}
	\label{lem:aux1}
    The following holds.
    \begin{enumerate}
    	\item \(\sum_{e \in O'\cap Y'_p} f(e\mid X_p)\leq \sum_{e \in O'\cap Y'_p} f(e\mid Y_e)\).
    	\item \(\sum_{e \in O'\cap X'_q} f(e\mid Y_q)\leq \sum_{e \in O'\cap X'_q} f(e\mid X_e)\).
    \end{enumerate}
\end{lemma}

\begin{proof}
    We prove only the first part of the lemma; the second follows by symmetry.
	By the definition of \(Y'_p\), for any \(e \in O'\cap Y'_p\), \(\delta_X=0\) when \(e\) was added to \(Y\).
    This means that \(\mathcal{C}=\{X,Y\}\) at that time.
    Since \(e\) was added to \(Y\) instead of \(X\), we have \(f(e\mid X_e)\leq f(e\mid Y_e)\).
    The lemma now follows by summing over \(O'\cap Y'_p\) and the submodularity of \(f\).
\end{proof}

\begin{lemma}
	\label{lem:not_added}
	The following holds.
	\begin{enumerate}
		\item If \(\delta_X=0\), then \(\sum_{e \in O'\setminus (X_{p} \cup Y'_p)} f(e\mid X_p)\leq \varepsilon \cdot f(O)\).
		\item If \(\delta_Y=0\), then \(\sum_{e \in O'\setminus (Y_q\cup X'_q)} f(e\mid Y_q)\leq \varepsilon \cdot f(O)\).
	\end{enumerate}
\end{lemma}

\begin{proof}
	We prove only the first part of the lemma; the second follows by symmetry.
    By the termination condition of the \emph{while} loop and the update rule of $\tau$, the last threshold \(\tau_{\mathrm{last}}\) satisfies \(\varepsilon(1-\varepsilon)\cdot \Gamma/B<\tau_{\mathrm{last}}\leq \varepsilon\cdot \Gamma/B\).

    Since \(\delta_X=0\), we have \(Y'_p=Y_q\), and
    \(\delta_X\) remained \(0\) throughout the last
    threshold scan.
    Every \(e\in O'\setminus(X_p\cup Y'_p)\) remained
    unselected during that scan.
    If its density into \(X\) met \(\tau_{\mathrm{last}}\)
    when it was examined, the algorithm would either add
    it to \(X\) or \(Y\), or set \(\delta_X=1\), because
    the same element is reconsidered for the remaining
    set after a failed insertion.
    Each possibility is impossible.
    By submodularity,
    \(f(e\mid X_p)/c(e)\leq\tau_{\mathrm{last}}
    \leq\varepsilon\Gamma/B\).
	By submodularity and the facts that \(c(O'\setminus (X_{p} \cup Y'_p))\leq B\) and \(\Gamma\leq f(O)\),
	\[ 
    \sum_{e \in O'\setminus (X_{p} \cup Y'_p)} f(e\mid X_p)\le c(O'\setminus (X_{p} \cup Y'_p))\cdot \varepsilon\cdot \Gamma/B\leq \varepsilon \cdot f(O).\qedhere
    \]
\end{proof}

\begin{lemma}
	\label{lem:added_is_larger}
	The following holds.
	\begin{enumerate}
		\item If \(\delta_X=1\), and it is set to \(1\) after the first iteration of the \emph{while} loop, then
        \[ (1-\varepsilon) \sum_{e \in O' \setminus (X_{p} \cup Y'_p)} f(e\mid X_{p})\leq (1+\varepsilon)\sum_{e \in (X_p\setminus O') \cup\{x_{p+1}\}} f(e\mid X_e) + \varepsilon \sum_{e \in X_{p+1} \cap O'} f(e\mid X_e). \]
		\item If \(\delta_Y=1\), and it is set to \(1\) after the first iteration of the \emph{while} loop, then
        \[ (1-\varepsilon) \sum_{e \in O' \setminus (Y_{q} \cup X'_q)} f(e\mid Y_{q})\leq (1+\varepsilon)\sum_{e \in (Y_q\setminus O') \cup\{y_{q+1}\}} f(e\mid Y_e) + \varepsilon \sum_{e \in Y_{q+1} \cap O'} f(e\mid Y_e). \]
	\end{enumerate}
\end{lemma}

\begin{proof}
	We prove only the first part of the lemma; the second follows by symmetry.
    Since \(\delta_X=1\), \(X_{p+1}\) is well-defined.
    We first consider \(O' \setminus (X_{p+1} \cup Y'_p)\).
    
    Let \(\tau_{x_{p+1}}\) denote the threshold at which the algorithm sets \(\delta_X=1\).
    Throughout the preceding threshold scan at
    \(\tau_{x_{p+1}}/(1-\varepsilon)\), we had
    \(\delta_X=0\), and every element of
    \(O'\setminus(X_{p+1}\cup Y'_p)\) remained unselected.
    If such an element met the density threshold for \(X\),
    the algorithm would either add it to \(X\) or \(Y\),
    or set \(\delta_X=1\), because the same element is
    reconsidered for the remaining set after a failed
    insertion.
    Each possibility contradicts the preceding observations.
    By submodularity, it follows that
	\[ \frac{f(e\mid X_p)}{c(e)}\leq  \frac{\tau_{x_{p+1}}}{1-\varepsilon}. \]
	Summing over \(O' \setminus (X_{p+1} \cup Y'_p)\) yields
	\begin{equation}
		(1-\varepsilon)\sum_{e \in O' \setminus (X_{p+1} \cup Y'_p)} f(e\mid  X_{p})\leq c(O' \setminus (X_{p+1} \cup Y'_p))\cdot \tau_{x_{p+1}}. \label{eq:not_added}
	\end{equation}
    
    On the other hand, we have the following claim.
    \begin{claim}
        \label{claim:threshold-greedy-gain}
        For any \(T\subseteq X_{p+1}\),
        \[ \sum_{e\in T}f(e\mid X_e)\geq c(T)\cdot \tau_{x_{p+1}}. \]
    \end{claim}
    \begin{proof}
        Every element inserted into \(X\) passed the density test
        at a threshold at least \(\tau_{x_{p+1}}\), and
        \(x_{p+1}\) passed the density test when the algorithm
        set \(\delta_X=1\). Therefore,
        \[
        \frac{f(e\mid X_e)}{c(e)}\geq \tau_{x_{p+1}}\qquad\text{for every }e\in X_{p+1}.
        \]
        Summing over \(T\) yields the claim.
    \end{proof}
    
    By Claim~\ref{claim:threshold-greedy-gain}, we have
    \begin{align*}
    &(1+\varepsilon)\sum_{e\in X_{p+1}\setminus O'}f(e\mid X_e)
    +\varepsilon\sum_{e\in X_{p+1}\cap O'}f(e\mid X_e)\\
    &\quad\geq
    \bigl((1+\varepsilon)c(X_{p+1}\setminus O')
    +\varepsilon c(X_{p+1}\cap O')\bigr)\tau_{x_{p+1}}.
    \end{align*}
    
    Since \(c(X_{p+1})> \varepsilon(1+\varepsilon)^i B\), we have
    \[ (1+\varepsilon)\cdot c(X_{p+1}) > \varepsilon(1+\varepsilon)^{i+1}B>B-c(r)\geq c(O'). \]
    Subtracting \(c(X_{p+1}\cap O')\) from both sides, we have
    \[ (1+\varepsilon)\cdot c(X_{p+1}\setminus O')+\varepsilon\cdot c(X_{p+1}\cap O')\geq c(O' \setminus X_{p+1}) \ge c(O' \setminus (X_{p+1} \cup Y'_p)). \]
    Combining these inequalities with Eq.~\eqref{eq:not_added}, we obtain
    \[ (1+\varepsilon)\sum_{e \in X_{p+1} \setminus O'} f(e\mid X_e)+\varepsilon \sum_{e \in X_{p+1} \cap O'} f(e\mid X_e) \geq  (1-\varepsilon)\sum_{e \in O' \setminus (X_{p+1} \cup Y'_p)} f(e\mid  X_{p}). \]
    
    If \(x_{p+1}\notin O'\), we have \(X_{p+1} \setminus O'=(X_{p} \setminus O')\cup\{x_{p+1}\}\) and \(O' \setminus (X_{p+1} \cup Y'_p)=O' \setminus (X_{p} \cup Y'_p)\).
    The lemma follows directly from the above inequality.
    
    If \(x_{p+1}\in O'\), we have \(X_{p+1}\setminus O'= X_p\setminus O'\).
    Then
    \begin{align*}
            & (1+\varepsilon)\sum_{e \in (X_p\setminus O') \cup\{x_{p+1}\}} f(e\mid X_e) + \varepsilon \sum_{e \in X_{p+1} \cap O'} f(e\mid X_e) \\
            =\ & (1+\varepsilon)\cdot f(x_{p+1}\mid X_p)+(1+\varepsilon)\sum_{e \in X_{p+1} \setminus O'} f(e\mid X_e) + \varepsilon \sum_{e \in X_{p+1} \cap O'} f(e\mid X_e) \\
            \geq\ & (1+\varepsilon)\cdot f(x_{p+1}\mid X_p)+ (1-\varepsilon) \sum_{e \in O' \setminus (X_{p+1} \cup Y'_p)} f(e\mid X_{p}) \\
            \geq\ & (1-\varepsilon) \sum_{e \in O' \setminus (X_{p} \cup Y'_p)} f(e\mid X_{p}).\qedhere
    \end{align*}
\end{proof}

We now bound \(f(O \cup S) - f(S + r)\) for \(S \in \{X_p, Y_q\}\).
When \(S\) is unsaturated (\(\delta_S = 0\)), Lemma~\ref{lem:marginal-gain-delta-zero0} applies;
when \(S\) is saturated (\(\delta_S = 1\)), Lemma~\ref{lem:marginal-gain-delta-zero1} applies.

\begin{lemma}
	\label{lem:marginal-gain-delta-zero0}
	The following holds.
	\begin{enumerate}
		\item  If \(\delta_X=0\), then 
        \[ f(O \cup X_p) - f(X_p+r)\leq \sum_{e \in O' \cap Y_q} f(e\mid Y_e) + \varepsilon\cdot f(O). \]
		\item If \(\delta_Y=0\), then 
        \[ f(O \cup Y_q) - f(Y_q+r)\leq \sum_{e \in O' \cap X_p} f(e\mid X_e) + \varepsilon\cdot f(O). \]
	\end{enumerate}
\end{lemma}

\begin{proof}
	We prove only the first part of the lemma; the second follows by symmetry.
	By Observation \ref{obs:S+r} and the submodularity of \(f\), we have
    \begin{align*}
		f(O\cup X_p)-f(X_p+r) &\leq f\left( O^{\prime} \cup X_p\right)-f\left( X_p\right) \\
		&\leq \sum_{e \in O^{\prime} \setminus X_p} f\left( e \mid X_p\right) \\
		&= \sum_{e \in O^{\prime} \cap Y'_p} f( e \mid X_p) + \sum_{e \in O^{\prime} \setminus\left(X_p \cup Y'_p\right)} f( e \mid X_p).
	\end{align*}
	Since \(\delta_X=0\), \(Y'_p=Y_q\) by definition.
    By Lemma~\ref{lem:aux1}, we have
	\[ \sum_{e \in O'\cap Y'_p} f(e\mid X_p)\leq \sum_{e \in O'\cap Y'_p} f(e\mid Y_e)=\sum_{e \in O' \cap Y_q} f(e\mid Y_e). \]
	By Lemma~\ref{lem:not_added}, we have
    \[ \sum_{e \in O'\setminus (X_{p} \cup Y'_p)} f(e\mid X_p)\leq \varepsilon \cdot f(O). \]
	Together, these inequalities imply the lemma.
\end{proof}

\begin{lemma}
	\label{lem:marginal-gain-delta-zero1}
	The following holds.
	\begin{enumerate}
		\item If \(\delta_X=1\), and it is set to \(1\) after the first iteration of the outer threshold loop, then
        \begin{align*}
            & f(O\cup X_p)-f(X_p+r) \\
            \leq\ &\sum_{e \in O^{\prime} \cap Y_q} f( e \mid Y_e)+\frac{1+\varepsilon}{1-\varepsilon}\sum_{e \in (X_p\setminus O') \cup\{x_{p+1}\}} f(e\mid X_e)+\frac{\varepsilon}{1-\varepsilon} \sum_{e \in X_{p+1} \cap O'} f(e\mid X_e).
        \end{align*}
		\item If \(\delta_Y=1\), and it is set to \(1\) after the first iteration of the outer threshold loop, then
        \begin{align*}
            & f(O\cup Y_q)-f(Y_q+r) \\ 
            \leq\ & \sum_{e \in O^{\prime} \cap X_p} f( e \mid X_e) +\frac{1+\varepsilon}{1-\varepsilon}\sum_{e \in (Y_q\setminus O') \cup\{y_{q+1}\}} f(e\mid Y_e)+\frac{\varepsilon}{1-\varepsilon} \sum_{e \in Y_{q+1} \cap O'} f(e\mid Y_e).
        \end{align*}
	\end{enumerate}
\end{lemma}

\begin{proof}
	We prove only the first part of the lemma; the second follows by symmetry.
	By Observation \ref{obs:S+r} and the submodularity of \(f\), we have
	\begin{align*}
		f(O\cup X_p)-f(X_p+r) &\leq f\left( O^{\prime} \cup X_p\right)-f\left( X_p\right) \\
		&\leq \sum_{e \in O^{\prime} \setminus X_p} f\left( e \mid X_p\right) \\
		&= \sum_{e \in O^{\prime} \cap Y'_p} f( e \mid X_p) + \sum_{e \in O^{\prime} \setminus\left(X_p \cup Y'_p\right)} f( e \mid X_p).
	\end{align*}
	By definition, \(Y'_p\subseteq Y_q\).
    Then, by Lemma~\ref{lem:aux1}, we have
	\begin{equation*}
		\sum_{e \in O^{\prime} \cap Y'_p} f\left( e \mid X_p\right)\leq \sum_{e \in O^{\prime} \cap Y'_p} f( e \mid Y_e)\leq \sum_{e \in O^{\prime} \cap Y_q} f( e \mid Y_e).
	\end{equation*}
	By Lemma~\ref{lem:added_is_larger}, we have
    \[ \sum_{e \in O' \setminus (X_{p} \cup Y'_p)} f(e\mid X_{p})\leq \frac{1+\varepsilon}{1-\varepsilon}\sum_{e \in (X_p\setminus O') \cup\{x_{p+1}\}} f(e\mid X_e)+\frac{\varepsilon}{1-\varepsilon} \sum_{e \in X_{p+1} \cap O'} f(e\mid X_e). \]
    Together, these inequalities imply the lemma.
\end{proof}

Finally, we establish the approximation ratio and query complexity of Algorithm~\ref{alg:small-r-fixed}.

\begin{theorem}
    Assume that \(c(r)\leq (1-\varepsilon)B\), Algorithm~\ref{alg:small-r-fixed} achieves a \((1/4-O(\varepsilon))\)-approximation and uses \(O(n\log^2(1/\varepsilon)/\varepsilon^2)\) queries.
\end{theorem}

\begin{proof}
    We first verify feasibility for every guess \(i\).
    The insertion rule maintains \(c(Z)\leq\varepsilon(1+\varepsilon)^iB\) for each \(Z\in\{X,Y\}\).
    Every \(e\in V'\) satisfies \(c(e)\leq(1-\varepsilon(1+\varepsilon)^i)B\), and therefore
    \[
    c(Z+e)\leq c(Z)+c(e)\leq B.
    \]
    Thus both augmented candidates are feasible, including the case where no element is added.
    Since the returned solution is the best of these candidates over all guesses and the initial empty set, it is feasible.
    
    We now analyze the approximation ratio for the correct guess.
    
	First, if \(\delta_X\) is set to \(1\) in the first iteration of the outer threshold loop, where \(\tau=19\Gamma/(4\varepsilon B)\),
    then
    \[ f(S^*)\geq f(X_{p+1})\geq c(X_{p+1})\cdot \tau\geq \varepsilon B\cdot 19\Gamma/(4\varepsilon B)\geq f(O)/4. \]
    The approximation guarantee holds.
    The same argument holds for \(Y_{q+1}\).
    Thus, in the following, if \(\delta_X\) or \(\delta_Y\) equals \(1\), we assume that they are set to \(1\) after the first iteration of the outer threshold loop.

    By Observation \ref{obs:upper-bound-of-f(O)},
    \begin{equation}
        f(O)-f(X_{p}+r)- f(Y_q+r) \leq  f(O\cup X_p)-f(X_p+r)+ f(O\cup Y_q)-f(Y_q+r). \label{eq:aux1}
    \end{equation}
    We prove the theorem with a case analysis.
    
	\textbf{Case 1.} \(\delta_X=\delta_Y=0\). By Eq.~\eqref{eq:aux1} and Lemma~\ref{lem:marginal-gain-delta-zero0}, we have
	\begin{align*}
		& f(O)-f(X_{p}+r)- f(Y_q+r) \\
        \leq\ &  \sum_{e \in O' \cap Y_q} f(e\mid Y_e) + \varepsilon\cdot f(O)+\sum_{e \in O' \cap X_p} f(e\mid X_e) + \varepsilon\cdot f(O) \\
		\leq\ &  f(X_p) + f(Y_q) + 2\varepsilon\cdot f(O).
	\end{align*}
    
    By rearranging the inequality, we have
    \[ f(S^*)\geq \frac{1}{4}(f(X_p)+f(Y_q)+f(X_p+r)+f(Y_q+r))\geq \frac{1-2\varepsilon}{4}\cdot f(O). \]

	\textbf{Case 2.} Exactly one of \(\delta_X\) and \(\delta_Y\) equals \(1\).
    Assume w.l.o.g.~that \(\delta_X=1\) and \(\delta_Y=0\).
    By Eq.~\eqref{eq:aux1} and Lemmas \ref{lem:marginal-gain-delta-zero0} and \ref{lem:marginal-gain-delta-zero1},
    \begin{align*}
       &\ f(O)-f(X_{p}+r)- f(Y_q+r) \\
       \leq &\ \sum_{e \in O^{\prime} \cap Y_q} f( e \mid Y_e)+\frac{1+\varepsilon}{1-\varepsilon}\sum_{e \in (X_p\setminus O') \cup\{x_{p+1}\}} f(e\mid X_e)+\frac{\varepsilon}{1-\varepsilon} \sum_{e \in X_{p+1} \cap O'} f(e\mid X_e) \\
       + &\ \sum_{e \in O' \cap X_p} f(e\mid X_e) + \varepsilon\cdot f(O) \\
       \leq &\ f(Y_q)+ \frac{1+2\varepsilon}{1-\varepsilon} \cdot f(X_{p+1}) + \varepsilon\cdot f(O).
    \end{align*}
    By rearranging the inequality, we have
    \[ \frac{1+2\varepsilon}{1-\varepsilon} \cdot f(S^*)\geq \frac{1}{4}\left(\frac{1+2\varepsilon}{1-\varepsilon} \cdot f(X_{p+1})+f(Y_q)+f(X_p+r)+f(Y_q+r)\right)\geq \frac{1-\varepsilon}{4}\cdot f(O). \]

    \textbf{Case 3.} \(\delta_X=\delta_Y=1\).
    By Eq.~\eqref{eq:aux1} and Lemma~\ref{lem:marginal-gain-delta-zero1},
    \begin{align*}
       &\ f(O)-f(X_{p}+r)- f(Y_q+r) \\
       \leq &\ \sum_{e \in O^{\prime} \cap Y_q} f( e \mid Y_e)+\frac{1+\varepsilon}{1-\varepsilon}\sum_{e \in (X_p\setminus O') \cup\{x_{p+1}\}} f(e\mid X_e)+\frac{\varepsilon}{1-\varepsilon} \sum_{e \in X_{p+1} \cap O'} f(e\mid X_e) \\
       + &\ \sum_{e \in O^{\prime} \cap X_p} f( e \mid X_e)+\frac{1+\varepsilon}{1-\varepsilon}\sum_{e \in (Y_q\setminus O') \cup\{y_{q+1}\}} f(e\mid Y_e)+\frac{\varepsilon}{1-\varepsilon} \sum_{e \in Y_{q+1} \cap O'} f(e\mid Y_e) \\
       \leq &\ \frac{1+2\varepsilon}{1-\varepsilon} \cdot (f(X_{p+1})+f(Y_{q+1})).
    \end{align*}
    By rearranging the inequality, we have
    \[ \frac{1+2\varepsilon}{1-\varepsilon} \cdot f(S^*)\geq \frac{1}{4}\left(\frac{1+2\varepsilon}{1-\varepsilon} \cdot (f(X_{p+1})+f(Y_{q+1}))+f(X_p+r)+f(Y_q+r)\right)\geq \frac{1}{4}\cdot f(O). \]
    
    To conclude, in all cases, \(f(S^*)\geq (1/4-\varepsilon)\cdot f(O)\).

    Now consider the query complexity of Algorithm~\ref{alg:small-r-fixed}.
    It invokes LA to get the value \(\Gamma\), which uses \(O(n)\) queries by Lemma~\ref{lem:LA-estimate}. 
    It has \(\left\lfloor\log_{1+\varepsilon}\varepsilon^{-1}\right\rfloor=O(\log(1/\varepsilon)/\varepsilon)\) guesses for the value of \(B-c(r)\).
    For each guess, the outer threshold loop has
    \(O(\log(1/\varepsilon)/\varepsilon)\) iterations.
    Each scan uses \(O(n)\) queries, since the inner retry loop
    considers at most two solutions per element.
    After the outer threshold loop, augmentation uses another
    \(O(n)\) queries.
    Thus, the algorithm uses \(O(n\log^2(1/\varepsilon)/\varepsilon^2)\) queries in total.
\end{proof}

\section{The Algorithm for \texorpdfstring{\(c(r)\ge (1-\varepsilon)B\)}{c(r) >= (1-epsilon)B}}
\label{sec:large_c(r)}

In this section, we present a linear-query deterministic algorithm for SMK that achieves a \((1/4-\varepsilon)\)-approximation, assuming that \(c(r)\geq (1-\varepsilon)B\).

Note that when \(f({r})\geq \frac{1}{4}\cdot f(O)\), we can achieve the desired approximation ratio by simply returning the best singleton using exactly \(n\) queries.
Thus, we assume that \(f({r})< \frac{1}{4}\cdot f(O)\).
By the submodularity and non-negativity of \(f\), \(f(O')=f(O-r)\geq f(O)-f(r)+f(\varnothing)>\frac{3}{4}\cdot f(O)\).
If we can find a set \(S\) that satisfies \(c(S)\leq B\) and \(f(S)\geq \frac{1}{3}\cdot f(O')\), then we obtain the desired approximation ratio.
Note that \(c(O')=c(O)-c(r)\leq B-(1-\varepsilon)B=\varepsilon B\) while our budget is \(B\).
This naturally leads to the study of \emph{bicriteria approximation algorithms} for SMK.
\begin{definition}[Bicriteria SMK]
    Given an SMK instance \(\max\{f(S):S\subseteq V,c(S)\leq B\}\), let \(O\) be its optimal solution.
    For any \(\alpha\in(0,1)\) and \(\beta\geq 1\), we say that an algorithm is a \((\alpha,\beta)\)-bicriteria approximation algorithm if it returns a set \(S\) such that \(f(S)\geq \alpha\cdot f(O)\) and \(c(S)\le \beta \cdot B\).
\end{definition}

For this problem, we obtain the following result.
{\ThmBicriteria*}

The proof of the theorem is presented in Section~\ref{subsec:bicriteria_approximation}. For now, we use it to solve SMK under the assumption that \(c(r)\geq (1-\varepsilon)B\).

\begin{theorem}
\label{thm:large-r-application}
Assume \(\varepsilon\in (0,1/12]\) and \(c(r)\ge (1-\varepsilon)B\).
There exists a deterministic algorithm for SMK that achieves a \((1/4-\varepsilon)\)-approximation and uses \(O(n\log(1/\varepsilon)/\varepsilon)\) queries.
\end{theorem}

\begin{proof}
The pseudo-code of the algorithm is presented below.
\begin{enumerate}
    \item Let \(e_{\max}\in\arg\max_{e\in V} f(e)\).
    \item Let \(S_b\) be a \((1/3-\varepsilon, 12)\)-bicriteria approximate solution on SMK instance 
    \begin{equation}
        \max\{f(S): c(S)\leq \varepsilon B\}. \label{eq:ins}
    \end{equation}
    \item Return \(S^*=\arg\max\{f(e_{\max}), f(S_b)\}\).
\end{enumerate}

We now give an analysis.
If \(f(r)\geq f(O)/4\), then
\[ f(S^*)\geq f(e_{\max})\geq f(r)\geq \frac{1}{4} \cdot f(O). \]
If \(f(r)<f(O)/4\), by the submodularity and non-negativity of \(f\),
\[ f(O')=f(O-r)\geq f(O)-f(r)+f(\varnothing)>\frac{3}{4}\cdot f(O). \]
Moreover, \(c(O')=c(O)-c(r)\leq \varepsilon B\), which means that \(O'\) is feasible for SMK instance \eqref{eq:ins}. Since \(S_b\) is a \((1/3-\varepsilon, 12)\)-bicriteria approximate solution for instance \eqref{eq:ins},
we have \(c(S_b)\leq 12\varepsilon B\leq B\) due to \(\varepsilon\leq 1/12\), and
\[ f(S^*)\geq f(S_b)\geq (1/3-\varepsilon)\cdot f(O')\geq (1/4-\varepsilon)\cdot f(O). \]

Finally, since it requires exactly \(n\) queries to find \(e_{\max}\) and \(O(n\log(1/\varepsilon)/\varepsilon)\) queries to find \(S_b\) by Theorem~\ref{thm:bicriteria-main}, the proposed algorithm uses \(O(n\log(1/\varepsilon)/\varepsilon)\) queries in total.
\end{proof}

\subsection{The Bicriteria Approximation}\label{subsec:bicriteria_approximation}

\begin{algorithm}[t]
    \SetAlgoLined
    \DontPrintSemicolon
    \caption{Repeated Threshold Greedy \(+\) Post-processing}
    \label{alg:bicriteria-general}
    \KwIn{Instance \((V,f,c,B)\), value \(\Gamma\) satisfying \(\Gamma\le f(O)\le 19\Gamma\), parameters \(\varepsilon\), \(\ell\) and \(B'\).}

    \(A\leftarrow\varnothing\).

    \For{\(i=1\) \KwTo \(\ell\)}{
        \(A_i\leftarrow\varnothing\).

        Let \(\tau\gets 19\Gamma/(2B')\).

        \While{\(\tau>\varepsilon(1-\varepsilon)\Gamma/B'\)}{
            \ForEach{\(e\in V\setminus A\)}{
                \If{\(c(A_i)<B'\) and \(\frac{f(e\mid A_i)}{c(e)}\geq \tau\)}{ \label{line:bicriteria-condition}
                    \(A_i\leftarrow A_i+e\).
                }
            }
            Update \(\tau\gets(1-\varepsilon)\tau\).
        }
        \(A\gets A\cup A_i\).
    }

    \For{\(i=1\) \KwTo \(\ell\)}{
        Define \(g_i(D)=f(A_i\cup D)\).

        \(D_i\leftarrow \text{DetUSM}(A, g_i,\varepsilon)\).
    }

    \Return the set in \(\{A_i\}_{i=1}^{\ell}\cup \{D_i\cup A_i\}_{i=1}^{\ell}\) that maximizes \(f\).

\end{algorithm}

In this section, we prove Theorem~\ref{thm:bicriteria-main} via Algorithm~\ref{alg:bicriteria-general}.
It takes two parameters \(\ell\) and \(B'\), which are set as a pair to match each part of the theorem.
The algorithm proceeds in two phases.
First, it repeatedly constructs pairwise-disjoint candidate sets \(A_1,\ldots, A_\ell\) using a threshold-greedy rule.
The threshold $\tau$ is initialized at \(19\Gamma/(2B')\) and decays geometrically by a factor of \((1-\varepsilon)\) until it falls below  \(\varepsilon(1-\varepsilon)\Gamma/B'\).
For each value of \(\tau\), the algorithm scans all the unselected elements.
As long as the cost of \(A_i\) is less than $B'$, elements whose marginal density exceeds \(\tau\) are added to $A_i$.
Second, for each \(A_i\), the algorithm refines it by solving an unconstrained submodular maximization problem over the ground set \(A = \bigcup_{i\in[\ell]}A_i\) with the objective function \(g_i(D) = f(A_i\cup D)\).
This is implemented by invoking the algorithm DetUSM from Lemma~\ref{lem:USM}.
Finally, the algorithm returns the best candidate among the sets \(\{A_i\}_{i\in[\ell]}\) and their refinements \(\{A_i\cup D_i\}_{i\in[\ell]}\).

We now state Theorem~\ref{thm:bicriteria-detail}, a detailed version of Theorem~\ref{thm:bicriteria-main}, whose proof is presented below.

\begin{restatable}{theorem}{ThmBicriteriaDetail}
    \label{thm:bicriteria-detail}
    Algorithm~\ref{alg:bicriteria-general} uses \(O(\ell n\log(1/\varepsilon)/\varepsilon)\) queries. Besides,
    \begin{enumerate}
        \item if \(\ell=6\) and \(B'=B\), Algorithm~\ref{alg:bicriteria-general} is a deterministic \((1/3-\varepsilon, 12)\)-bicriteria approximation algorithm for SMK that uses \(O(n\log(1/\varepsilon)/\varepsilon)\) queries;
        \item if \(\ell=\lceil 1/\varepsilon\rceil\) and \(B'=2B\), Algorithm~\ref{alg:bicriteria-general} is a deterministic \((1/2-\varepsilon, O(1/\varepsilon))\)-bicriteria approximation algorithm for SMK that uses \(O(n\log(1/\varepsilon)/\varepsilon^2)\) queries.
    \end{enumerate}
\end{restatable}

We now present the analysis.
Let \(S^*\) denote the set returned by Algorithm~\ref{alg:bicriteria-general}.
Observe that \(A_i\cap A_j=\varnothing\) for all \(i\neq j\in[\ell]\), and \(A=\cup_{i\in[\ell]}A_i\).
As usual, let \(O\) be an optimal solution, and let \(\hat{O} \coloneqq O\setminus A\) and \(\dot{O}\coloneqq O\cap A\).

\begin{lemma}\label{lem:bicriteria-marginal}
 Assume that the \emph{while} loop did not break in the first iteration.
    Then for all \(i\in[\ell]\),
    \[
        f(\hat{O}\cup A_i)-f(A_i)\leq\frac{B}{B'}\cdot\frac{f(A_i)}{1-\varepsilon}+\varepsilon\cdot f(O).
    \]
\end{lemma}

\begin{proof}
    \textbf{Case 1.} \(c(A_i)\geq B'\).
    Let \(\tau^*\) denote the largest threshold for which \(c(A_i)\leq B'\). Then, in the last iteration where \(\tau=(1-\varepsilon)^{-1}\tau^*\), we have $c(A_i)<B'$ and no element of \(\hat{O}\) was added to $A_i$.
    For each $e\in \hat{O}$, it was not added because its marginal density is below the threshold $(1-\varepsilon)^{-1}\tau^*$.
    By submodularity, it follows that
    \[
        \frac{f(e\mid A_i)}{c(e)}\leq \frac{\tau^*}{1-\varepsilon},
    \]
    and
    \[
        f(\hat{O}\cup A_i)-f(A_i)\leq \sum_{e\in\hat{O}} f(e\mid A_i)\leq c(\hat{O})\cdot \frac{\tau^*}{1-\varepsilon}\leq \frac{\tau^*}{1-\varepsilon}\cdot B.
    \]
    On the other hand, the marginal density of every element in \(A_i\) is at least \(\tau^*\).
    Together with condition \(c(A_i)\geq B'\), we have
    \[
        f(A_i)\geq c(A_i)\cdot \tau^*\geq \tau^*B'.
    \]
    Combining the last two inequalities yields
    \[
        f(\hat{O}\cup A_i)-f(A_i)\leq \frac{B}{B'}\cdot\frac{f(A_i)}{1-\varepsilon}.
    \]

    \textbf{Case 2.} \(c(A_i)<B'\).
    By the \emph{while} loop condition of Algorithm~\ref{alg:bicriteria-general}, the last threshold \(\tau_{\mathrm{last}}\) satisfies \(\varepsilon(1-\varepsilon)\Gamma/B'<\tau_{\mathrm{last}}\leq \varepsilon\cdot \Gamma/B'\). For every \(e\in\hat{O}\), since \(c(A_i)<B'\), it was not added to \(A_i\) because its marginal density is less than $\tau_{\mathrm{last}}$, i.e.,
    \[ \frac{f(e\mid A_{i})}{c(e)}\leq \tau_{\mathrm{last}}\leq \varepsilon\cdot \Gamma/B'. \]
    By submodularity and the fact that $\Gamma\leq f(O)$,
    \[ f(\hat{O}\cup A_i)-f(A_i)\leq \sum_{e\in\hat{O}} f(e\mid A_i)\leq c(\hat{O})\cdot \varepsilon\cdot \Gamma/B'\leq \varepsilon \cdot f(O). \]
    
    Combining both cases, we obtain
    \[
        f(\hat{O}\cup A_i)-f(A_i)\leq\max\left\{\frac{B}{B'}\cdot\frac{f(A_i)}{1-\varepsilon},\varepsilon \cdot f(O)\right\}\leq\frac{B}{B'}\cdot\frac{f(A_i)}{1-\varepsilon}+\varepsilon\cdot f(O).\qedhere
    \]
\end{proof}

We are now ready to prove Theorem~\ref{thm:bicriteria-detail}.

\begin{proof}
    If the \emph{while} loop terminates in its first iteration, when \(\tau = 19\Gamma/(2B')\), then \(c(A_i) \geq B'\), and for \(B'\in \{B,2B\}\),
    \[ f(A_i)\geq c(A_i)\cdot\tau\geq B'\cdot 19\Gamma/(2B')\geq f(O)/2. \]
    The approximation guarantee holds for both claims.
    Thus, we assume that the \emph{while} loop terminates after its first iteration in the following.

    Since the \(A_i\)'s are pairwise disjoint, repeated application of submodularity gives
    \[ \sum_{i \in [\ell]} f(O \cup A_i) \geq (\ell - 1)\cdot f(O) + f(O \cup A) \geq (\ell - 1)\cdot f(O). \]
    Thus, by an averaging argument, there exists \(i\in[\ell]\) such that \(f(O\cup A_i)\geq \left(1-\frac1\ell\right)\cdot f(O)\).

    Since $\dot{O}\subseteq A$, by Lemma~\ref{lem:USM},
    \begin{equation}\label{eq:USM}
        \begin{split}
            f(D_i \cup A_i)= g_i(D_i) &\geq \left(\frac{1}{2}-\varepsilon\right)\cdot g_i(\dot{O}) + \frac{1}{4}\cdot g_i(\varnothing) \\
            &= \left(\frac{1}{2}-\varepsilon\right)\cdot f(\dot{O} \cup A_i)
            + \frac{1}{4}\cdot f(A_i).
        \end{split}
    \end{equation}
    Then, we have
    \begin{align*}
        \left(1 - \frac{1}{\ell}\right)\cdot f(O) &\leq f(O \cup A_i) \\
        &\leq f(\hat{O} \cup A_i) - f(A_i) + f(\dot{O} \cup A_i) \\
        &\leq \frac{B}{B'}\cdot\frac{f(A_i)}{1-\varepsilon}+\varepsilon\cdot f(O) + \left(\frac{2}{1-2\varepsilon}\cdot f(D_i \cup A_i) - \frac{f(A_i)}{2-4\varepsilon}\right).
        \end{align*}
    The second inequality follows from submodularity, and the third inequality follows from Lemma~\ref{lem:bicriteria-marginal} and Eq.~\eqref{eq:USM}.
    
    If \(\ell=6\) and \(B'=B\), the above inequality reduces to
    \[ \frac{5}{6}\cdot f(O)\leq \frac{2}{1-2\varepsilon}\cdot f(D_i \cup A_i)+ \frac{1-3\varepsilon}{(1-\varepsilon)(2-4\varepsilon)}\cdot f(A_i)+\varepsilon\cdot f(O)\leq \frac{5-7\varepsilon}{2-6\varepsilon}\cdot f(S^*)+\varepsilon\cdot f(O). \]
    By rearranging the inequality, we obtain
    \[ 
    f(S^*)\geq \left(\frac13-\varepsilon\right)\cdot\frac{5-6\varepsilon}{5-7\varepsilon}f(O)\geq\left(\frac13-\varepsilon\right)\cdot f(O).
    \]
    
    If \(\ell=\lceil 1/\varepsilon\rceil\) and \(B'=2B\),
    the above inequality reduces to 
    \[ (1-\varepsilon)\cdot f(O)\leq \frac{2}{1-2\varepsilon}\cdot f(D_i \cup A_i)+\varepsilon\cdot f(O)\leq \frac{2}{1-2\varepsilon}\cdot f(S^*)+\varepsilon\cdot f(O). \]
    By rearranging the inequality, we obtain
    \[ f(S^*)\geq \frac{\left(1-2\varepsilon\right)^2}{2}\cdot f(O)\ge\left(\frac{1}{2}-2\varepsilon\right)\cdot f(O). \]

    For each \(i\in [\ell]\), since Algorithm~\ref{alg:bicriteria-general} adds elements to \(A_i\) only when \(c(A_i)<B'\), we have \(c(A_i)<B'+B\).
    Since \(A=\cup_{i=1}^{\ell} A_i\) and \(A_i\)'s are pairwise disjoint, \(c(A)=\sum_{i=1}^{\ell}c(A_i)\leq \ell(B'+B)\).
    Note that \(S^*\) is a subset of \(A\), thus \(c(S^*)\leq \ell(B'+B)\).
    Consequently, when \(\ell=6\) and \(B'=B\), we have \(c(S^*)\leq 12B\); when \(\ell=\lceil 1/\varepsilon\rceil\) and \(B'=2B\), we have \(c(S^*)=O(\varepsilon^{-1}B)\).

    It remains to bound the number of queries used by Algorithm~\ref{alg:bicriteria-general}.
    It invokes LA to get value \(\Gamma\), which uses \(O(n)\) queries by Lemma~\ref{lem:LA-estimate}.
    It then repeatedly constructs \(\ell\) disjoint \(A_i\)'s.
    For each \(A_i\), it performs a \emph{while} loop with \(O(\log(1/\varepsilon)/\varepsilon)\) iterations.
    In each iteration, it scans the elements and uses \(O(n)\) queries.
    Thus, the construction of \(A_i\)'s uses \(O(\ell n\log(1/\varepsilon)/\varepsilon)\) queries.
    The post-processing phase invokes \textsc{DetUSM} \(\ell\) times, which uses \(O(\ell n/\varepsilon)\) additional queries by Lemma~\ref{lem:USM}.
    To summarize, the total number of queries is \(O(\ell n\log(1/\varepsilon)/\varepsilon)\).
    Substituting \(\ell=6\) and \(\ell=\lceil 1/\varepsilon\rceil\) gives the correct number of queries, respectively.
\end{proof}

\section{Conclusion}
\label{sec:conclusion}

In this paper, we presented a deterministic linear-query algorithm for non-monotone SMK that achieves a $(1/4-\varepsilon)$-approximation, matching the best previously known randomized ratio.
Our result is obtained by combining the threshold-twin-greedy framework with residual-budget enumeration, and observing a novel bicriteria reduction.
As a secondary contribution, we obtained a $(1/2-\varepsilon, O(1/\varepsilon))$-bicriteria approximation with linear query complexity, improving upon prior quadratic-query algorithms.

A natural open question is whether any linear-query algorithm, deterministic or randomized, can break the $1/4$ barrier.

\bibliography{SMK}

@inproceedings{AmanatidisFLLR20,
  author       = {Georgios Amanatidis and
                  Federico Fusco and
                  Philip Lazos and
                  Stefano Leonardi and
                  Rebecca Reiffenh{\"{a}}user},
  title        = {Fast Adaptive Non-Monotone Submodular Maximization Subject to a Knapsack
                  Constraint},
  booktitle    = {Advances in Neural Information Processing Systems 33: Annual Conference
                  on Neural Information Processing Systems 2020, NeurIPS 2020, December
                  6-12, 2020, virtual},
  year         = {2020},
  bibsource    = {dblp computer science bibliography, https://dblp.org}
}

@inproceedings{BadanidiyuruV14,
  author       = {Ashwinkumar Badanidiyuru and
                  Jan Vondr{\'{a}}k},
  title        = {Fast algorithms for maximizing submodular functions},
  booktitle    = {Proceedings of the Twenty-Fifth Annual {ACM-SIAM} Symposium on Discrete
                  Algorithms, {SODA} 2014, Portland, Oregon, USA, January 5-7, 2014},
  pages        = {1497--1514},
  publisher    = {{SIAM}},
  year         = {2014},
  bibsource    = {dblp computer science bibliography, https://dblp.org}
}

@article{BuchbinderF18,
  author       = {Niv Buchbinder and
                  Moran Feldman},
  title        = {Deterministic Algorithms for Submodular Maximization Problems},
  journal      = {{ACM} Trans. Algorithms},
  volume       = {14},
  number       = {3},
  pages        = {32:1--32:20},
  year         = {2018},
  bibsource    = {dblp computer science bibliography, https://dblp.org}
}

@article{BuchbinderF19,
  author       = {Niv Buchbinder and
                  Moran Feldman},
  title        = {Constrained Submodular Maximization via a Nonsymmetric Technique},
  journal      = {Math. Oper. Res.},
  volume       = {44},
  number       = {3},
  pages        = {988--1005},
  year         = {2019},
  bibsource    = {dblp computer science bibliography, https://dblp.org}
}

@inproceedings{BuchbinderF19soda,
  author       = {Niv Buchbinder and
                  Moran Feldman and
                  Mohit Garg},
  editor       = {Timothy M. Chan},
  title        = {Deterministic ({\textonehalf} + {\(\epsilon\)})-Approximation for
                  Submodular Maximization over a Matroid},
  booktitle    = {Proceedings of the Thirtieth Annual {ACM-SIAM} Symposium on Discrete
                  Algorithms, {SODA} 2019, San Diego, California, USA, January 6-9,
                  2019},
  pages        = {241--254},
  publisher    = {{SIAM}},
  year         = {2019},
  bibsource    = {dblp computer science bibliography, https://dblp.org}
}

@inproceedings{BuchbinderF24,
  author       = {Niv Buchbinder and
                  Moran Feldman},
  title        = {Constrained Submodular Maximization via New Bounds for DR-Submodular
                  Functions},
  booktitle    = {Proceedings of the 56th Annual {ACM} Symposium on Theory of Computing,
                  {STOC} 2024, Vancouver, BC, Canada, June 24-28, 2024},
  pages        = {1820--1831},
  publisher    = {{ACM}},
  year         = {2024},
  bibsource    = {dblp computer science bibliography, https://dblp.org}
}

@inproceedings{BuchbinderF24focs,
  author       = {Niv Buchbinder and
                  Moran Feldman},
  title        = {Deterministic Algorithm and Faster Algorithm for Submodular Maximization
                  Subject to a Matroid Constraint},
  booktitle    = {65th {IEEE} Annual Symposium on Foundations of Computer Science, {FOCS}
                  2024, Chicago, IL, USA, October 27-30, 2024},
  pages        = {700--712},
  publisher    = {{IEEE}},
  year         = {2024},
  bibsource    = {dblp computer science bibliography, https://dblp.org}
}

@inproceedings{BuchbinderF25,
  author       = {Niv Buchbinder and
                  Moran Feldman},
  editor       = {Michal Kouck{\'{y}} and
                  Nikhil Bansal},
  title        = {Extending the Extension: Deterministic Algorithm for Non-monotone
                  Submodular Maximization},
  booktitle    = {Proceedings of the 57th Annual {ACM} Symposium on Theory of Computing,
                  {STOC} 2025, Prague, Czechia, June 23-27, 2025},
  pages        = {1130--1141},
  publisher    = {{ACM}},
  year         = {2025},
  bibsource    = {dblp computer science bibliography, https://dblp.org}
}

@inproceedings{Crawford23,
  author       = {Victoria G. Crawford},
  title        = {Scalable Bicriteria Algorithms for Non-Monotone Submodular Cover},
  booktitle    = {International Conference on Artificial Intelligence and Statistics,
                  25-27 April 2023, Palau de Congressos, Valencia, Spain},
  series       = {Proceedings of Machine Learning Research},
  volume       = {206},
  pages        = {9517--9537},
  publisher    = {{PMLR}},
  year         = {2023},
  bibsource    = {dblp computer science bibliography, https://dblp.org}
}

@inproceedings{EneN19,
  author       = {Alina Ene and
                  Huy L. Nguyen},
  title        = {A Nearly-Linear Time Algorithm for Submodular Maximization with a
                  Knapsack Constraint},
  booktitle    = {46th International Colloquium on Automata, Languages, and Programming,
                  {ICALP} 2019, Patras, Greece, July 9-12, 2019},
  series       = {LIPIcs},
  volume       = {132},
  pages        = {53:1--53:12},
  publisher    = {Schloss Dagstuhl - Leibniz-Zentrum f{\"{u}}r Informatik},
  year         = {2019},
  bibsource    = {dblp computer science bibliography, https://dblp.org}
}

@article{FeldmanK25,
  author       = {Moran Feldman and
                  Alan Kuhnle},
  title        = {Bicriteria Submodular Maximization},
  journal      = {CoRR},
  volume       = {abs/2507.10248},
  year         = {2025},
  eprinttype   = {arXiv},
  eprint       = {2507.10248},
  bibsource    = {dblp computer science bibliography, https://dblp.org}
}

@article{FeldmanNS23,
  author       = {Moran Feldman and
                  Zeev Nutov and
                  Elad Shoham},
  title        = {Practical Budgeted Submodular Maximization},
  journal      = {Algorithmica},
  volume       = {85},
  number       = {5},
  pages        = {1332--1371},
  year         = {2023},
  bibsource    = {dblp computer science bibliography, https://dblp.org}
}

@inproceedings{GuptaRST10,
  author       = {Anupam Gupta and
                  Aaron Roth and
                  Grant Schoenebeck and
                  Kunal Talwar},
  title        = {Constrained Non-monotone Submodular Maximization: Offline and Secretary
                  Algorithms},
  booktitle    = {Internet and Network Economics - 6th International Workshop, {WINE}
                  2010, Stanford, CA, USA, December 13-17, 2010. Proceedings},
  series       = {Lecture Notes in Computer Science},
  volume       = {6484},
  pages        = {246--257},
  publisher    = {Springer},
  year         = {2010},
  bibsource    = {dblp computer science bibliography, https://dblp.org}
}

@article{HanCZZWYXTH21,
  author       = {Kai Han and
                  Shuang Cui and
                  Tianshuai Zhu and
                  Enpei Zhang and
                  Benwei Wu and
                  Zhizhuo Yin and
                  Tong Xu and
                  Shaojie Tang and
                  He Huang},
  title        = {Approximation Algorithms for Submodular Data Summarization with a
                  Knapsack Constraint},
  journal      = {Proc. {ACM} Meas. Anal. Comput. Syst.},
  volume       = {5},
  number       = {1},
  pages        = {05:1--05:31},
  year         = {2021},
  bibsource    = {dblp computer science bibliography, https://dblp.org}
}

@inproceedings{KempeKT03,
  author       = {David Kempe and
                  Jon M. Kleinberg and
                  {\'{E}}va Tardos},
  title        = {Maximizing the spread of influence through a social network},
  booktitle    = {Proceedings of the Ninth {ACM} {SIGKDD} International Conference on
                  Knowledge Discovery and Data Mining, Washington, DC, USA, August 24
                  - 27, 2003},
  pages        = {137--146},
  publisher    = {{ACM}},
  year         = {2003},
  bibsource    = {dblp computer science bibliography, https://dblp.org}
}

@article{KhullerMN99,
  author       = {Samir Khuller and
                  Anna Moss and
                  Joseph Naor},
  title        = {The Budgeted Maximum Coverage Problem},
  journal      = {Inf. Process. Lett.},
  volume       = {70},
  number       = {1},
  pages        = {39--45},
  year         = {1999},
  bibsource    = {dblp computer science bibliography, https://dblp.org}
}

@inproceedings{Kuhnle21,
  author       = {Alan Kuhnle},
  title        = {Nearly Linear-Time, Parallelizable Algorithms for Non-Monotone Submodular
                  Maximization},
  booktitle    = {Thirty-Fifth {AAAI} Conference on Artificial Intelligence, {AAAI}
                  2021, Thirty-Third Conference on Innovative Applications of Artificial
                  Intelligence, {IAAI} 2021, The Eleventh Symposium on Educational Advances
                  in Artificial Intelligence, {EAAI} 2021, Virtual Event, February 2-9,
                  2021},
  pages        = {8200--8208},
  publisher    = {{AAAI} Press},
  year         = {2021},
  bibsource    = {dblp computer science bibliography, https://dblp.org}
}

@article{KulikST13,
  author       = {Ariel Kulik and
                  Hadas Shachnai and
                  Tami Tamir},
  title        = {Approximations for Monotone and Nonmonotone Submodular Maximization
                  with Knapsack Constraints},
  journal      = {Math. Oper. Res.},
  volume       = {38},
  number       = {4},
  pages        = {729--739},
  year         = {2013},
  bibsource    = {dblp computer science bibliography, https://dblp.org}
}

@inproceedings{LeskovecKGFVG07,
  author       = {Jure Leskovec and
                  Andreas Krause and
                  Carlos Guestrin and
                  Christos Faloutsos and
                  Jeanne M. VanBriesen and
                  Natalie S. Glance},
  title        = {Cost-effective outbreak detection in networks},
  booktitle    = {Proceedings of the 13th {ACM} {SIGKDD} International Conference on
                  Knowledge Discovery and Data Mining, San Jose, California, USA, August
                  12-15, 2007},
  pages        = {420--429},
  publisher    = {{ACM}},
  year         = {2007},
  bibsource    = {dblp computer science bibliography, https://dblp.org}
}

@inproceedings{LiF0K22,
  author       = {Wenxin Li and
                  Moran Feldman and
                  Ehsan Kazemi and
                  Amin Karbasi},
  title        = {Submodular Maximization in Clean Linear Time},
  booktitle    = {Advances in Neural Information Processing Systems 35: Annual Conference
                  on Neural Information Processing Systems 2022, NeurIPS 2022, New Orleans,
                  LA, USA, November 28 - December 9, 2022},
  year         = {2022},
  bibsource    = {dblp computer science bibliography, https://dblp.org}
}

@inproceedings{LinB11,
  author       = {Hui Lin and
                  Jeff A. Bilmes},
  title        = {A Class of Submodular Functions for Document Summarization},
  booktitle    = {The 49th Annual Meeting of the Association for Computational Linguistics:
                  Human Language Technologies, Proceedings of the Conference, 19-24
                  June, 2011, Portland, Oregon, {USA}},
  pages        = {510--520},
  publisher    = {The Association for Computer Linguistics},
  year         = {2011},
  bibsource    = {dblp computer science bibliography, https://dblp.org}
}

@article{MirzasoleimanKS16,
  author       = {Baharan Mirzasoleiman and
                  Amin Karbasi and
                  Rik Sarkar and
                  Andreas Krause},
  title        = {Distributed Submodular Maximization},
  journal      = {J. Mach. Learn. Res.},
  volume       = {17},
  pages        = {238:1--238:44},
  year         = {2016},
  bibsource    = {dblp computer science bibliography, https://dblp.org}
}

@inproceedings{NarasimhanB05,
  author       = {Mukund Narasimhan and
                  Jeff A. Bilmes},
  title        = {A Submodular-supermodular Procedure with Applications to Discriminative
                  Structure Learning},
  booktitle    = {{UAI} '05, Proceedings of the 21st Conference in Uncertainty in Artificial
                  Intelligence, Edinburgh, Scotland, July 26-29, 2005},
  pages        = {404--412},
  publisher    = {{AUAI} Press},
  year         = {2005},
  bibsource    = {dblp computer science bibliography, https://dblp.org}
}

@article{NemhauserW78,
  author       = {George L. Nemhauser and
                  Laurence A. Wolsey},
  title        = {Best Algorithms for Approximating the Maximum of a Submodular Set
                  Function},
  journal      = {Math. Oper. Res.},
  volume       = {3},
  number       = {3},
  pages        = {177--188},
  year         = {1978},
  bibsource    = {dblp computer science bibliography, https://dblp.org}
}

@article{Pham25,
  author       = {Canh V. Pham},
  title        = {Enhanced deterministic approximation algorithm for non-monotone submodular
                  maximization under knapsack constraint with linear query complexity},
  journal      = {J. Comb. Optim.},
  volume       = {49},
  number       = {1},
  pages        = {2},
  year         = {2025},
  bibsource    = {dblp computer science bibliography, https://dblp.org}
}

@inproceedings{PhamT0T23,
  author       = {Canh V. Pham and
                  Tan D. Tran and
                  Dung K. T. Ha and
                  My T. Thai},
  title        = {Linear Query Approximation Algorithms for Non-monotone Submodular
                  Maximization under Knapsack Constraint},
  booktitle    = {Proceedings of the Thirty-Second International Joint Conference on
                  Artificial Intelligence, {IJCAI} 2023, 19th-25th August 2023, Macao,
                  SAR, China},
  pages        = {4127--4135},
  publisher    = {ijcai.org},
  year         = {2023},
  bibsource    = {dblp computer science bibliography, https://dblp.org}
}

@inproceedings{Qi22,
  author       = {Benjamin Qi},
  title        = {On Maximizing Sums of Non-Monotone Submodular and Linear Functions},
  booktitle    = {33rd International Symposium on Algorithms and Computation, {ISAAC}
                  2022, Seoul, South Korea, December 19-21, 2022},
  series       = {LIPIcs},
  volume       = {248},
  pages        = {41:1--41:16},
  publisher    = {Schloss Dagstuhl - Leibniz-Zentrum f{\"{u}}r Informatik},
  year         = {2022},
  bibsource    = {dblp computer science bibliography, https://dblp.org}
}

@article{Sviridenko04,
  author       = {Maxim Sviridenko},
  title        = {A note on maximizing a submodular set function subject to a knapsack
                  constraint},
  journal      = {Oper. Res. Lett.},
  volume       = {32},
  number       = {1},
  pages        = {41--43},
  year         = {2004},
  bibsource    = {dblp computer science bibliography, https://dblp.org}
}

@article{Wolsey82,
  author       = {Laurence A. Wolsey},
  title        = {An analysis of the greedy algorithm for the submodular set covering
                  problem},
  journal      = {Comb.},
  volume       = {2},
  number       = {4},
  pages        = {385--393},
  year         = {1982},
  bibsource    = {dblp computer science bibliography, https://dblp.org}
}

\end{document}